\documentclass[a4paper, authorcolumns]{lipics-v2021}
\usepackage{mathtools}
\usepackage{cite}
\usepackage{float}
\usepackage{graphicx}
\nolinenumbers

\graphicspath{ {./images/} }

\newcommand{\ST}{\,:\,}                 

\newcommand{\bdOmega}{\partial \kern+1pt \Omega} 

\newcommand{\SP}{\kern+1pt}             
\newcommand{\revFunk}[1]{{}^r \kern-1pt F_{#1}}
\DeclareMathOperator{\interior}{int}

\usepackage{xcolor}

\title{Visualizing Higher Order Structures, Overlap Regions, and Clustering in the Hilbert Geometry.}
\titlerunning{\small Visualizing Higher Order Structures, Overlap Regions, and Clustering in the Hilbert Geometry.}

\author{Hridhaan Banerjee}{ Thomas Jefferson High School for Science and Technology, Virginia, USA \and \url{~}}{hridhaan.s.banerjee@gmail.com}{}{}

\author{Soren Brown}{Department of Computer Science, University of Maryland, College Park, USA  \and \url{~}}{sorenkbrown@gmail.com}{}{}

\author{June Cagan}{Department of Computer Science, University of Maryland, College Park, USA  \and \url{~}}{jcagan@terpmail.umd.edu}{}{}

\author{Auguste H. Gezalyan}{Université de Lorraine, CNRS, Inria, LORIA, F-54000 Nancy, France \and \url{~}}{octavo@umd.edu}{https://orcid.org/0000-0002-5704-312X}{}

\author{Megan Hunleth}{Montgomery Blair High School, Silver Spring, Maryland, USA  \and \url{~}}{megan@hunleth.com}{}{}

\author{Veena Kailad}{Montgomery Blair High School, Silver Spring, Maryland, USA \and \url{~}}{veena.s.kailad@gmail.com}{}{}

\author{Chaewoon Kyoung}{Montgomery Blair High School, Silver Spring, Maryland, USA \and \url{~}}{kchwoony210@gmail.com}{}{}

\author{Rowan Shigeno}{Department of Mathematics, Haverford College, Pennsylvania, USA\and \url{~}}{rowanshigeno@gmail.com}{}{}

\author{Yasmine Tajeddin}{Department of Computer Science, University of Maryland, College Park, USA \and \url{~}}{Yasminet@umd.edu}{}{}

\author{Andrew Wagger}{Department of Computer Science, University of Maryland, College Park, USA \and \url{https://andrew.wagger.net}}{andwagger@gmail.com}{}{}

\author{Kelin Zhu}{Department of Mathematics, University of Maryland, College Park, USA \and \url{~}}{kelinzhu@terpmail.umd.edu}{}{}

\author{David M. Mount}{Department of Computer Science, University of Maryland, College Park, Maryland, USA \and \url{https://www.cs.umd.edu/~mount/}}{mount@umd.edu}{https://orcid.org/0000-0002-3290-8932}{}

\authorrunning{\small Banerjee, Brown, Cagan, Gezalyan, Hunleth, Kailad, Kyoung, Shigeno, Tajeddin, Wagger, Zhu}

\Copyright{Banerjee, Brown, Cagan, Gezalyan, Hunleth, Kailad, Kyoung, Shigeno, Tajeddin, Wagger, Zhu}

\ccsdesc[500]{Theory of computation~Computational geometry}
\keywords{Hilbert metric, Funk metric, Voronoi diagrams}

\date{\today}

\EventEditors{Hee-Kap Ahn, Michael Hoffmann, and Amir Nayyeri}
\EventNoEds{3}
\EventLongTitle{42nd International Symposium on Computational Geometry (SoCG 2026)}
\EventShortTitle{SoCG 2026}
\EventAcronym{SoCG}
\EventYear{2026}
\EventDate{June 2--5, 2026}
\EventLocation{New Brunswick, NJ, USA}
\EventLogo{socg-logo}
\def\EventLogoHeight{42}
\SeriesVolume{367}
\ArticleNo{108}  

\category{Media Exposition}

\funding{REU-CAAR 2025}
\begin{document}

\supplementdetails[subcategory={Software}]{Software}{https://andrew.wagger.net/hilbert/}
\supplementdetails[subcategory={Video}]{Media}{https://youtu.be/yhJn--3Qgks}

\maketitle

\begin{abstract}
Higher-order Voronoi diagrams and Delaunay mosaics in polygonal metrics have only recently been studied, yet no tools exist for visualizing them. We introduce a tool that fills this gap, providing dynamic interactive software for visualizing higher-order Voronoi diagrams and Delaunay mosaics along with clustering and tools for exploring overlap and outer regions in the Hilbert polygonal metric. We prove that $k^{th}$ order Voronoi cells are not always star-shaped and establish complexity bounds for our algorithm, which generates all order Voronoi diagrams at once. Our software unifies and extends previous tools for visualizing the Hilbert, Funk, and Thompson geometries. 
\end{abstract}

\section{Introduction}
The Hilbert metric has attracted recent interest due to its applications in clustering \cite{nielsen2019clustering}, non-linear embeddings \cite{nielsen2023non}, convex approximation \cite{abdelkader2018delone, abdelkader2024convex}, and real analysis \cite{lemmens2014birkhoff}. As such, there has been work on extending results from computational geometry to the Hilbert metric, such as Voronoi diagrams (including farthest-point Voronoi)\cite{gezalyan2023voronoi,bumpus2023software, song2025farthest}, minimum enclosing balls \cite{banerjee2024heine}, and Delaunay triangulations \cite{gezalyan2023delaunay}. To aid in these endeavors, software to help visualize the Hilbert metric has been progressively developed. This began with a Hilbert ball visualizer \cite{nielsen2017balls}, then some Voronoi demo software \cite{bumpus2023software}, which evolved into software for the Funk and Thomson polygonal geometry \cite{banerjeesoftware2025}. We build on this foundation by providing the first software with interactive support for dynamically visualizing higher-order Voronoi diagrams and Delaunay mosaics in the Hilbert metric, exploring overlap and outer regions (defined in Section 3), and step-by-step clustering on arbitrary convex polygonal domains. The software is implemented in JavaScript, it runs interactively in the browser, and it includes a tutorial video and brute force Voronoi mode for verification. It also maintains the features from all previous software.

\section{Higher order structures}
The study of $k^{th}$ order abstract Voronoi diagrams started with ``On the complexity of higher order abstract Voronoi diagrams'' where cells were determined by the $k$ closest sites \cite{BOHLER2015539}. The Hilbert metric does not immediately fall into the category of $k^{th}$ order abstract Voronoi diagrams, but can be made to fit it by extending bisectors to infinity without the extensions crossing. As such, there exist frameworks for efficient $k^{th}$ order Voronoi diagrams \cite{bohler2019efficient,BOHLER2015539}. Our visualization generates all orders at once using an edge labeling technique similar to those used in Euclidean geometry \cite{claverol2024edge}.  For completeness, we begin this section by proving two supporting lemmas that show our algorithm works. Then we present the algorithm. Throughout, we assume that we have been given some convex polygonal region $\Omega$ with $m$ sides and a set of $n$ sites $S$. Both proofs use a classical technique using properties of balls \cite{lee1982k}.

\begin{figure}[h]
\centering
\includegraphics[scale=.55]{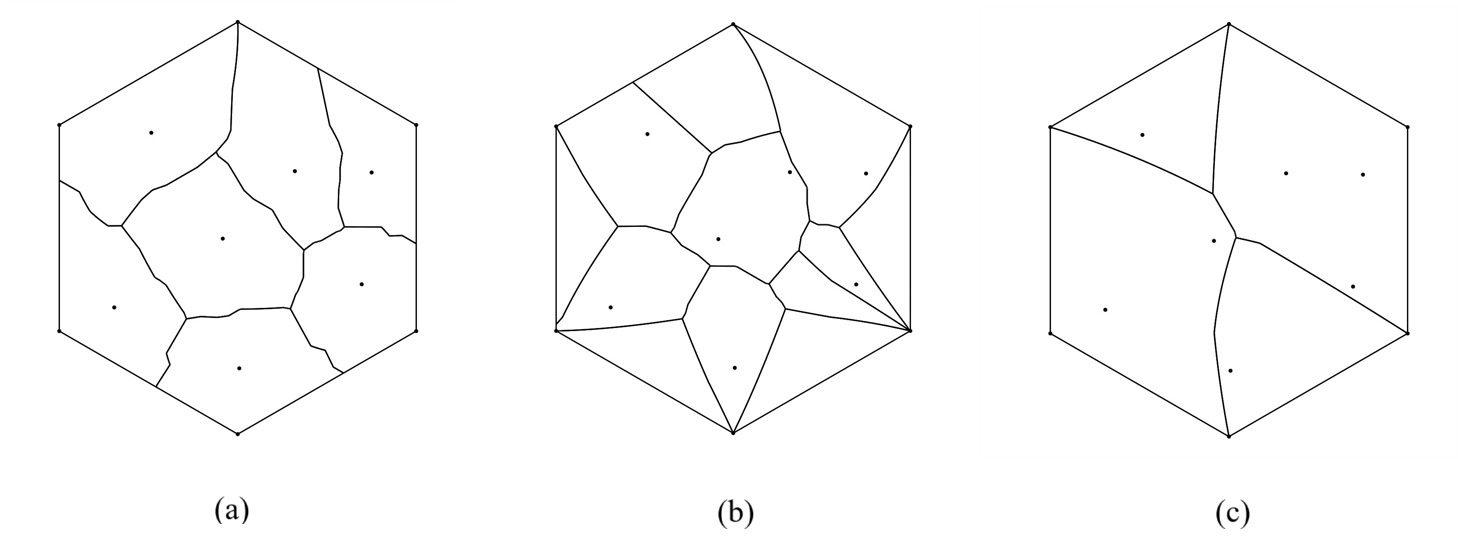}
  \caption{(a) First order Hilbert Voronoi (b) second order (c) final ($(n-1)^{th}$) order}
  \label{fig:4.5Parallel}
\end{figure}

\begin{lemma} \label{lem:vertex}
The circumcenter of any three sites serves as a Voronoi vertex in both the $k^{th}$ and $(k+1)^{th}$ order Hilbert Voronoi diagrams for some $0<k<n-1$.
\end{lemma}

\begin{proof}
Consider the circumcenter of three sites $s_1,s_2,s_3$ lying on the boundary of ball $B$, with $k-1$ sites of $S$ in the ball. Perturbing $B$ towards one of the sites will change which of $s_1,s_2,s_3$ is the $k^{th}$ closest site to the center of $B$. Perturbing and growing the ball along one of the bisectors to include two sites will change which site is the $(k+1)^{st}$ closest.
\end{proof}

\begin{lemma}\label{singleusage}
Every bisector path between circumcenters belongs to exactly one $k^{th}$ order cell.
\end{lemma}

\begin{proof}
Consider the bisector between two sites $s_1, s_2 \in S$ lying on a $k^{th}$-order cell. A ball moving along the bisector and passing through both sites contains exactly $k-1$ other sites before reaching a circumcenter with a third site. Perturbing the ball off the bisector causes it to include either $s_1$ or $s_2$, shifting the center into the $k^{{th}}$ order cell of one or the other.
\end{proof}

Our algorithm proceeds in three phases: Preprocessing computes all pairwise bisectors and all three site circumcenters parametrized along their respective bisectors. Looping over the bisectors lets us determine the diagram containment of the first bisector portion by checking the distance to each site. Traversing the bisector and updating diagram containment at each circumcenter yields a labeling of each bisector portion by the degree-order Hilbert Voronoi diagram it is part of, allowing us to display it quickly based on user input.

In the following we implicitly use three facts from \cite{gezalyan2023delaunay,gezalyan2023voronoi}; first, we can order the boundary of $\Omega$ to calculate Hilbert distances in time $O(\log m)$, second we can order the boundary of $\Omega$ to calculate circumcenters in time $O(\log^3 m)$, and third, bisectors have complexity $O(m)$.

\begin{corollary}
Our algorithm takes time $O(n^2 m \log m + n^3 \log^3 m + n^3 \log n \log m)$ with space $O(m\log m + n^3 + n^2 m)$. 
\end{corollary}
\begin{proof}
Preprocessing takes $O(n^2m \log m)$ time to compute all bisectors and $O(n^3 \log^3 m)$ time for computing all circumcenters (which we bucket by bisector). Sorting the circumcenters along the bisectors in the buckets takes $O(n^3 \log n)$. 

We determine the containment of the first bisector component in $O(n\log n \log m)$ time per bisector. Then, we traverse the bisector, determining the containment of each vertex and each bisector segment in $O(n)$ time, since there are $O(n)$ circumcenters along each bisector. Since we do this for all $n^2$ bisectors, this takes $O(n^2 (n\log n \log m + n))$. Our final runtime is $O(n^2 m \log m + n^3 \log^3 m + n^3 \log n \log m)$.

The storage consists of $O(m\log m)$ to order the boundary, $O(n^3)$ for circumcenters ordered along bisectors, and $O(n^2m)$ for the bisectors, giving us $O(m\log m + n^3 + n^2 m)$ storage.
\end{proof}

\begin{observation}
In the Hilbert metric, $k^{th}$ order Voronoi cells are not always star-shaped, a star-shaped region has an interior point from which all others are visible via a segment.
\end{observation}
\begin{proof}
 This can be seen with the following $\Omega$:  a square with coordinates from (100, 100) to (300, 300) and sites at (160, 284.9), (140, 170), (130, 165), (180, 285). The second-order Voronoi generates a non-star-shaped region. This configuration is in general position.
\end{proof}

\begin{figure}[h]
    \centering
    \includegraphics[scale=.32]{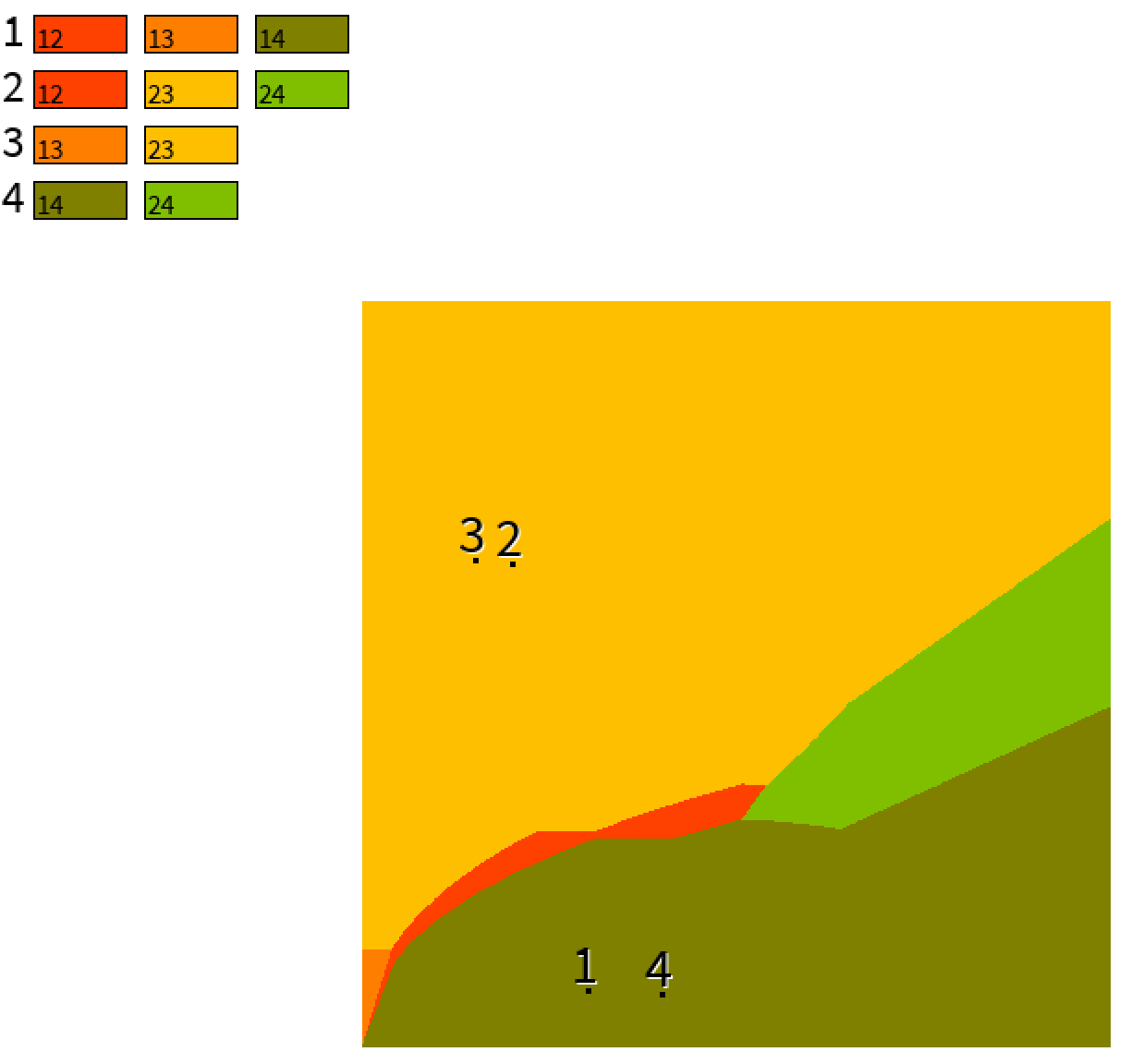}
    \caption{The red region, the Voronoi cell closest to 1 and 2, is not a star-shaped region.}
\end{figure}

We leverage our results on $k^{th}$-order cells to compute the Delaunay mosaic by computing a modified Fréchet mean (a point that minimizes the sum of distances) of the sites of each cell, and connecting them if their cells are adjacent.

\section{Overlap and Outer Region} \label{sec:overlap}
One notable property of the Hilbert metric is that three sites do not always have a circumcenter \cite{bumpus2023software}. First we recall what it means to be a ball at infinity through two sites \cite{gezalyan2023delaunay}. Then we define the overlap and outer regions which characterize circumcenter existence.

\begin{definition}[Infinite Ball]
    Given two sites $s_1, s_2 \in S$, consider balls passing through them with centers along the bisector parametrized by $t \in [0,1]$. This yields two limit balls on the boundary of $\Omega$: $B_0(p \ST q)$ at $t = 0$ and $B_1(p \ST q)$ at $t = 1$. (See Figure~\ref{fig:infinite}.)
\end{definition}

\begin{figure}[h]
    \centering
    \includegraphics[scale=.28]{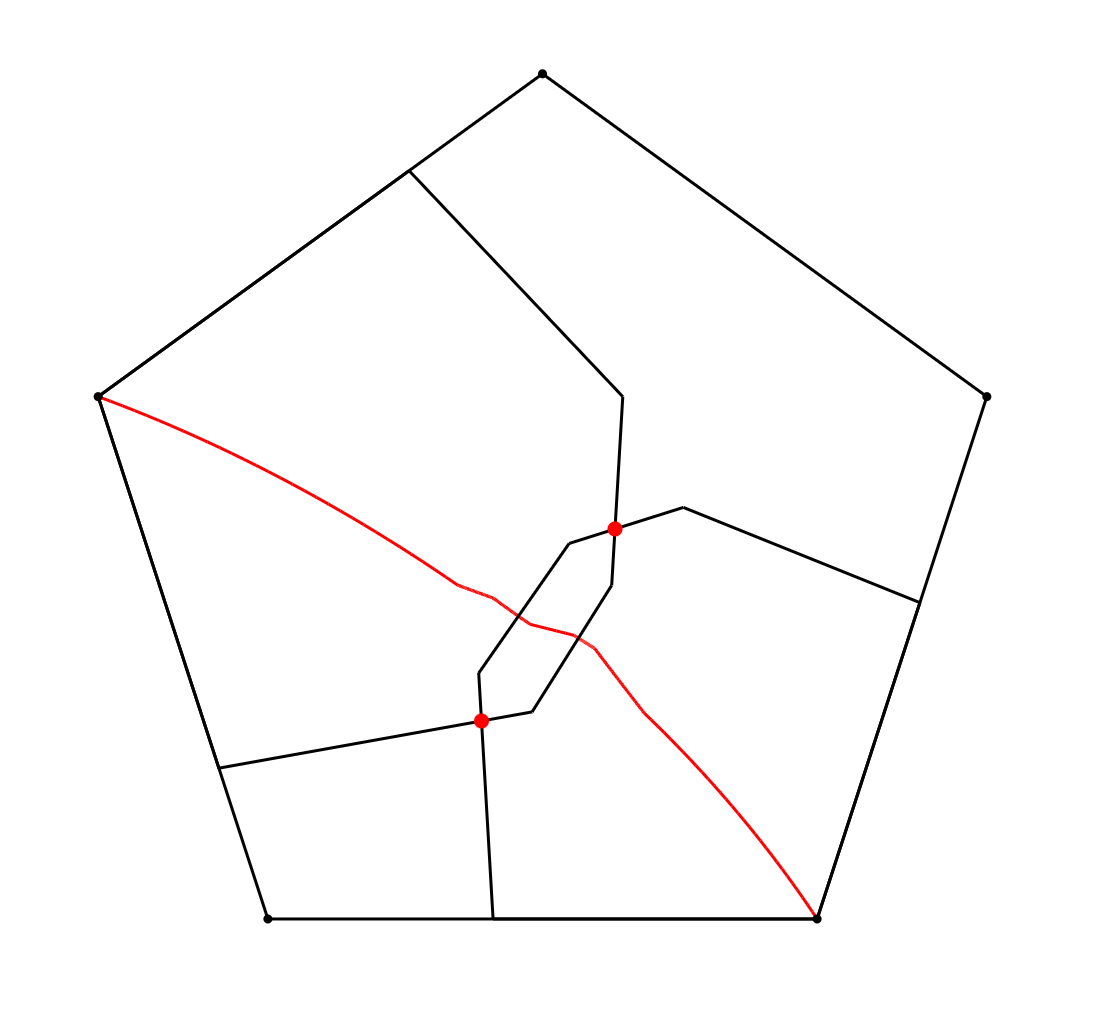}
    \caption{Two infinite balls along bisector through two sites (in red) generated by our software.}
    \label{fig:infinite}
\end{figure}

\begin{definition}[Overlap Region]
    Given two sites $s_1,s_2\in S$, their overlap region, denoted $Z(p,q)$, is $B_0(p \ST q)\cap B_1(p \ST q)$.
\end{definition}

\begin{definition}[Outer Region]
    Given two sites $s_1,s_2\in S$, their outer region, denoted $W(p,q)$, is $\interior \Omega - \left(B_0(p \ST q)\cup B_1(p \ST q)\right)$.
\end{definition}

The \emph{overlap} is the region between two sites $s_1$, $s_2$ where there is no circumcenter between $s_1$, $s_2$, and $p$ \cite{gezalyan2023delaunay}. The \emph{outer region} is the region that is not between the two sites where there is no circumcenter. In our software, we give users the option to display these regions dynamically. Users can generate these regions for any two sites, then move a third site between them to observe the behavior of the bisectors of the three sites.

\section{Clustering}
In 2019, Nielsen and Sun introduced clustering in the Hilbert simplex \cite{nielsen2019clustering}. We expand on this by allowing users to interactively visualize clustering on arbitrary Hilbert polygonal geometries. We provide $k$-means and single linkage clustering (see Figure~\ref{fig:clusters}).

\begin{figure}[h]
\centering
\includegraphics[scale=.45]{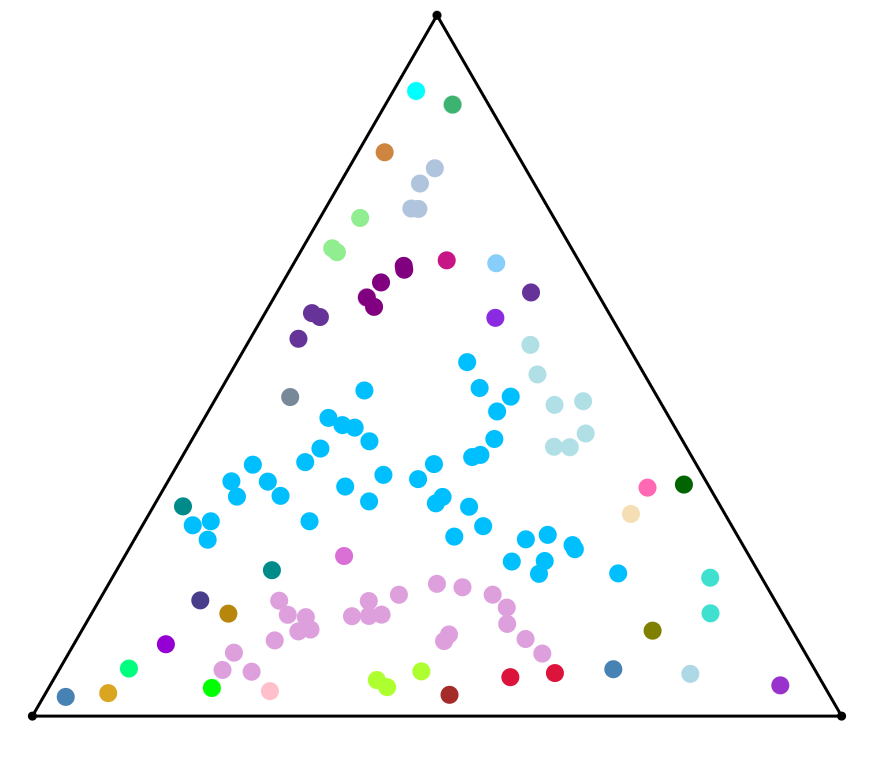}
  \caption{120 points in 38 clusters using single link clustering.}
  \label{fig:clusters}
\end{figure}

\section{Conclusion}
We have introduced the first software for visualizing higher-order Voronoi diagrams, Delaunay mosaics, clustering, overlap and outer regions in arbitrary convex Hilbert polygonal geometries. 

\bibliography{shortcuts,hilbert}

@inproceedings{abdelkader2018delone,
  author = {Abdelkader, Ahmed and Mount, David M.},
  title = {Economical {Delone} Sets for Approximating Convex Bodies},
  booktitle = SWAT_2018,
  year = {2018},
  pages = {4:1--4:12}, 
  doi = {10.4230/LIPIcs.SWAT.2018.4},
}

@inproceedings{nielsen2017balls,
  title = {On balls in a {Hilbert} polygonal geometry (multimedia contribution)},
  author = {Nielsen, Frank and Shao, Laetitia},
  booktitle = SOCG_2017,
  series = {Leibniz International Proceedings in Informatics (LIPIcs)},
  volume = {77},
  year = {2017},
  pages = {67:1--67:4},
  publisher = {Schloss Dagstuhl--Leibniz-Zentrum f{\"u}r Informatik},
  doi = {10.4230/LIPIcs.SoCG.2017.67},
}

@incollection{nielsen2019clustering,
  title = {Clustering in {Hilbert's} Projective Geometry: {The} Case Studies of the Probability Simplex and the Elliptope of Correlation Matrices},
  author = {Nielsen, Frank and Sun, Ke},
  booktitle = {Geometric Structures of Information},
  editor = {Nielsen, Frank},
  publisher = {Springer Internat.\ Pub.},
  year = {2019},
  pages = {297--331},
  doi = {10.1007/978-3-030-02520-5_11},
}

@article{banerjee2024heine,
  title={On The Heine-Borel Property and Minimum Enclosing Balls},
  author={Banerjee, Hridhaan and Day, Carmen Isabel and Hunleth, Megan and Hwang, Sarah and Gezalyan, Auguste H and Golovatskaia, Olya and Parepally, Nithin and Wang, Lucy and Mount, David M},
  journal={arXiv preprint arXiv:2412.17138},
  year={2024}
}

@inproceedings{nielsen2023non,
  title={Non-linear embeddings in {Hilbert} simplex geometry},
  author={Nielsen, Frank and Sun, Ke},
  booktitle={Topological, Algebraic and Geometric Learning Workshops 2023},
  pages={254--266},
  year={2023},
  organization={PMLR}
}

@inproceedings{abdelkader2024convex,
  title={Convex Approximation and the {Hilbert} Geometry},
  author={Abdelkader, Ahmed and Mount, David M.},
  booktitle={2024 Symposium on Simplicity in Algorithms (SOSA)},
  pages={286--298},
  year={2024},
  organization={SIAM},
  doi = {10.1137/1.9781611977936.26},
}

@inproceedings{gezalyan2023voronoi,
  title={Voronoi Diagrams in the {Hilbert} Metric},
  author={Gezalyan, Auguste H. and Mount, David M.},
  booktitle={39th International Symposium on Computational Geometry (SoCG 2023)},
  year={2023},
  organization={Schloss Dagstuhl-Leibniz-Zentrum f{\"u}r Informatik},
  doi={10.4230/LIPIcs.SoCG.2023.35},
}

@article{bumpus2023software,
  title={Analysis of Dynamic Voronoi Diagrams in the Hilbert Metric},
  author={Bumpus, Madeline and Dai, Caesar and Gezalyan, Auguste H. and Munoz, Sam and Santhoshkumar, Renita and Ye, Songyu and Mount, David M.},
  journal={Proceedings of the 35th Canadian Conference on Computational
Geometry (CCCG 2023)
Montreal, Canada},
  year={2023}
}

@InProceedings{gezalyan2023delaunay,
  author =	{Gezalyan, Auguste H. and Kim, Soo H. and Lopez, Carlos and Skora, Daniel and Stefankovic, Zofia and Mount, David M.},
  title =	{{Delaunay Triangulations in the Hilbert Metric}},
  booktitle =	{19th Scandinavian Symposium and Workshops on Algorithm Theory (SWAT 2024)},
  pages =	{25:1--25:17},
  series =	{Leibniz International Proceedings in Informatics (LIPIcs)},
  ISBN =	{978-3-95977-318-8},
  ISSN =	{1868-8969},
  year =	{2024},
  volume =	{294},
  editor =	{Bodlaender, Hans L.},
  publisher =	{Schloss Dagstuhl -- Leibniz-Zentrum f{\"u}r Informatik},
  address =	{Dagstuhl, Germany},
  URL =		{https://drops.dagstuhl.de/entities/document/10.4230/LIPIcs.SWAT.2024.25},
  URN =		{urn:nbn:de:0030-drops-200657},
  doi =		{10.4230/LIPIcs.SWAT.2024.25}
}

@misc{lemmens2014birkhoff, title={Birkhoff’s version of Hilbert’s metric and its applications in analysis}, url={http://dx.doi.org/10.4171/147-1/10}, DOI={10.4171/147-1/10}, journal={IRMA Lectures in Mathematics and Theoretical Physics}, publisher={EMS Press}, author={Lemmens, Bas and Nussbaum, Roger D.}, year={2014}, month=dec, pages={275–303}, language={en} }

@inproceedings{song2025farthest,
  author    = {Minju Song, Mook Kwon Jung and Hee-Kap Ah},
  title     = {Farthest-point Voronoi Diagrams in the Hilbert Metric},
  booktitle = {Proceedings of the 19th Algorithms and Data Structures Symposium (WADS 2025)},
  year      = {2025},
  location  = {York University, Toronto, Canada},
}

@misc{banerjeesoftware2025,
   title = {{Software For the Thompson and Funk Geometry}}, 
   author = {Banerjee, Hridhaan and Day, Carmen Isabel and Gezalyan, Auguste H. and Golovatskaia, Olga and Hunleth, Megan and Hwang, Sarah and Parepally, Nithin and Wang, Lucy and Mount, David M.},
   url = {https://github.com/nithin1527/funk-geo-visualizer},
   doi = {10.4230/artifacts.23291},
}

@article{BOHLER2015539,
title = {On the complexity of higher order abstract Voronoi diagrams},
journal = {Computational Geometry},
volume = {48},
number = {8},
pages = {539-551},
year = {2015},
issn = {0925-7721},
doi = {https://doi.org/10.1016/j.comgeo.2015.04.008},
url = {https://www.sciencedirect.com/science/article/pii/S0925772115000346},
author = {Cecilia Bohler and Panagiotis Cheilaris and Rolf Klein and Chih-Hung Liu and Evanthia Papadopoulou and Maksym Zavershynskyi},
}

@article{bohler2019efficient,
  title={An efficient randomized algorithm for higher-order abstract Voronoi diagrams},
  author={Bohler, Cecilia and Klein, Rolf and Liu, Chih-Hung},
  journal={Algorithmica},
  volume={81},
  number={6},
  pages={2317--2345},
  year={2019},
  publisher={Springer}
}

@article{claverol2024edge,
  title={The edge labeling of higher order Voronoi diagrams},
  author={Claverol, Merc{\`e} and de las Heras Parrilla, Andrea and Huemer, Clemens and Mart{\'\i}nez-Moraian, Alejandra},
  journal={Journal of Global Optimization},
  volume={90},
  number={2},
  pages={515--549},
  year={2024},
  publisher={Springer}
}

@article{lee1982k,
  title={On k-nearest neighbor Voronoi diagrams in the plane},
  author={Lee, Der-Tsai},
  journal={IEEE transactions on computers},
  volume={100},
  number={6},
  pages={478--487},
  year={1982},
  publisher={IEEE}
}

@string{SOCG_2017 = "Proc.\ 33rd Internat.\ Sympos.\ Comput.\ Geom."}

@string{SWAT_2018 = "Proc.\ 16th Scand.\ Workshop Algorithm Theory"}
\end{document}